\documentclass[10pt]{article}

\usepackage[a4paper,top=2cm,bottom=2cm,left=3cm,right=3cm,marginparwidth=1.75cm]{geometry}

\usepackage{amsmath}
\usepackage{amssymb}
\usepackage{amsthm}
\usepackage{graphicx}
\usepackage{hyperref}
\usepackage[all]{xy}
\usepackage{enumitem}
\usepackage{authblk}
\usepackage{xcolor}
\usepackage{textcomp}

\graphicspath{{././figure/}}

\theoremstyle{plain}
\newtheorem{theorem}{Theorem}[section]

\newtheorem{proposition}[theorem]{Proposition}

\theoremstyle{definition}

\newtheorem{example}[theorem]{Example}
\theoremstyle{remark}

\numberwithin{equation}{section}

\newcommand{\CC}{\mathbb{C}}
\newcommand{\NN}{\mathbb{N}}
\newcommand{\PP}{\mathbb{P}}
\newcommand{\RR}{\mathbb{R}}
\newcommand{\ZZ}{\mathbb{Z}}
\newcommand{\id}{\mathrm{id}}

\renewcommand{\Im}{\operatorname{Im}}
\newcommand{\Int}{\operatorname{Int}}

\newcommand{\Lk}{\mathrm{Lk}}

\title{Analysis of Hopf solitons as generalized fold maps}
\author[1,2]{Yuta Nozaki}
\author[3]{Darian Hall}
\author[2,3,4,5]{Ivan I. Smalyukh}
\author[2,6,*]{Yuya Koda}
\affil[1]{\footnotesize Department of Mathematics, Faculty of Science, Hokkaido University, Sapporo 060-0810, Japan;}
\affil[2]{\footnotesize 
International Institute for Sustainability with Knotted Chiral Meta Matter (WPI-SKCM$^2$), Hiroshima University, 1-3-1 Kagamiyama, Higashi-Hiroshima, Hiroshima 739-8531, Japan;}
\affil[3]{\footnotesize 
Department of Physics and Chemical Physics Program,
University of Colorado, Boulder, CO 80309, USA;}
\affil[4]{\footnotesize 
Department of Electrical, Computer, and Energy Engineering,
Materials Science and Engineering Program, University of Colorado, Boulder, CO 80309, USA;}
\affil[5]{\footnotesize 
Renewable and Sustainable Energy Institute, National Renewable Energy Laboratory,
University of Colorado, Boulder, CO 80309, USA;}
\affil[6]{\footnotesize Department of Mathematics, Hiyoshi Campus, Keio University, Yokohama 223-8521, Japan; and}
\affil[*]{\footnotesize Corresponding author: \texttt{koda@keio.jp}}

\begin{document}
\date{}
\maketitle

\begin{abstract}
The Hopf index, a topological invariant that quantifies the linking of preimage fibers, is fundamental to the structure and stability of hopfions. In this work, we propose a new mathematical framework for modeling hopfions with high Hopf index, drawing on the language of singularity theory and the topology of differentiable maps. At the core of our approach is the notion of a generalized Hopf map of order $n$, whose structure is captured via fold maps and their Stein factorizations. We demonstrate that this theoretical construction not only aligns closely with recent experimental observations of high-Hopf-index hopfions, but also offers a precise classification of the possible configurations of fiber pairs associated to distinct points. Our results thus establish a robust bridge between the geometry of singular maps and the experimentally observed topology of complex field configurations of hopfions in materials and other physical systems.
\end{abstract}

\section{Introduction}
\label{sec:introduction}

The \emph{Hopf map}, also known as the \emph{Hopf fibration},  
is a map $\varphi$ from the three-dimensional sphere $S^3$ to the two-dimensional sphere $S^2$ that provides a fundamental example of a homotopically non-trivial map, 
representing a generator of $\pi_3 (S^2) \cong \ZZ$. 
It is specifically defined as 
$\varphi (z_1, z_2 ) =  ( z_1 : z_2 )$, 
where $S^3$ is regarded as 
$\{ (z_1, z_2) \in \CC^2 \mid |z_1|^2 + |z_2|^2 = 1 \} \subset \CC^2$, $S^2$ is identified with the complex projective line $\CC \PP^1$, and $(z_1 : z_2)$ is the homogeneous coordinates of $\CC \PP^1$. 
See Figure~\ref{fig:Hopf_fibration}.
This map is characterized by the fact that the preimage of any two distinct points in $S^2$ 
forms linked circles known as the \emph{Hopf link}, which is illustrated in Figure~\ref{fig:Hopf_link}. 

\begin{figure}[htbp]
\centering\includegraphics[width=0.5\textwidth]{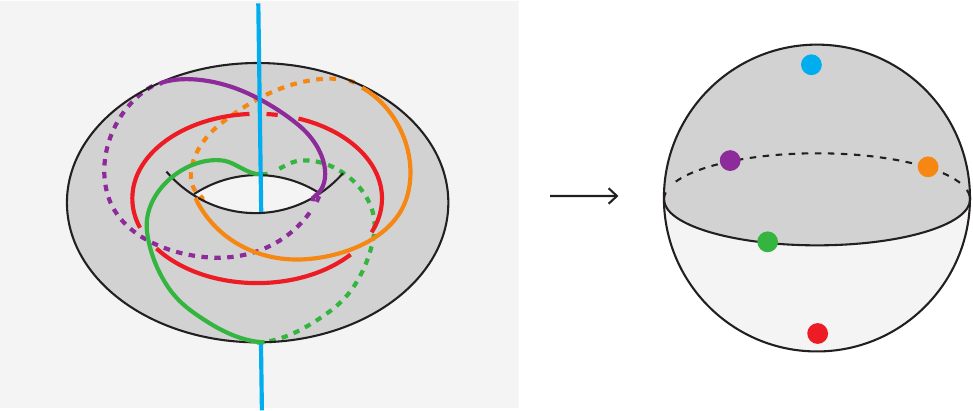}
\begin{picture}(400,0)(0,0)
\put(218,66){$\varphi$}
\put(145,-1){$S^3$}
\put(268,-1){$S^2$}
\end{picture}
\caption{The Hopf map.}
\label{fig:Hopf_fibration}
\end{figure} 

\begin{figure}[htbp]
\centering\includegraphics[width=0.16\textwidth]{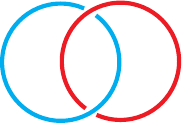}
\begin{picture}(400,0)(0,0)
\end{picture}
\caption{The Hopf link.}
\label{fig:Hopf_link}
\end{figure}

Since its discovery, the Hopf fibration has played a crucial role in diverse mathematical and physical contexts, ranging from homotopy groups of spheres to applications in describing various physical systems, ranging from subatomic particles to liquid crystals.
We refer the reader to \cite{Shn18}, \cite{MaSu04}.

Hopfions, topological solitons associated with Hopf fibrations, have been studied extensively in mathematical physics, theories of particle physics and cosmology and material science (see \cite{FaNi97}, \cite{Han17}, \cite{AcSm17PRX}, \cite{AcSm17NM}, \cite{TaSm18}). 
They appear in diverse physical systems such as liquid crystals, magnetism, and field theory, where they represent stable, topologically non-trivial configurations. 
The Hopf index, which quantifies the linking of preimage fibers, plays a crucial role in understanding the structure and stability of hopfions. 
While most theoretical and experimental studies focused on elementary hopfions with Hopf index $Q=1$ or $Q=-1$, recent experiments \cite{TAS18} also uncovered interesting examples of hopfions with integer Hopf index invariants different from unity, as well as the energetically stabilized by medium's chirality composite states of hopfions with opposite values of Hopf index that have the net Hopf index $Q=0$. Additionally, interesting examples of non-axisymmetric hopfions include one that can feature single preimages shaped as Hopf links or unlinked loops (see \cite[Figure~6]{TAS18}).
In addition, chiral liquid crystals with helical background structure were found hosting 2D and 3D crystals of hopfions of the so-called ``heliknoton'' type (see \cite{TaSm19}, \cite{VTS20}, \cite{TWS22}).
These heliknotons are found to have Hopf indices different from $Q=1$ or $Q=-1$, with the diverse types of experimentally and numerically reconstructed preimage linkings for one such heliknoton of $Q=2$ illustrated in Figure~\ref{fig:heliknoton_motivation}. Similar diversity for many other experimentally and computationally observed complex and high-integer heliknotons is discussed elsewhere \cite{HTKNS}, calling for a detailed mathematical analysis of possible types of preimage linkings in such topological objects, motivating our present study.

\begin{figure}[htbp]
\centering
\includegraphics[width=0.9\textwidth]{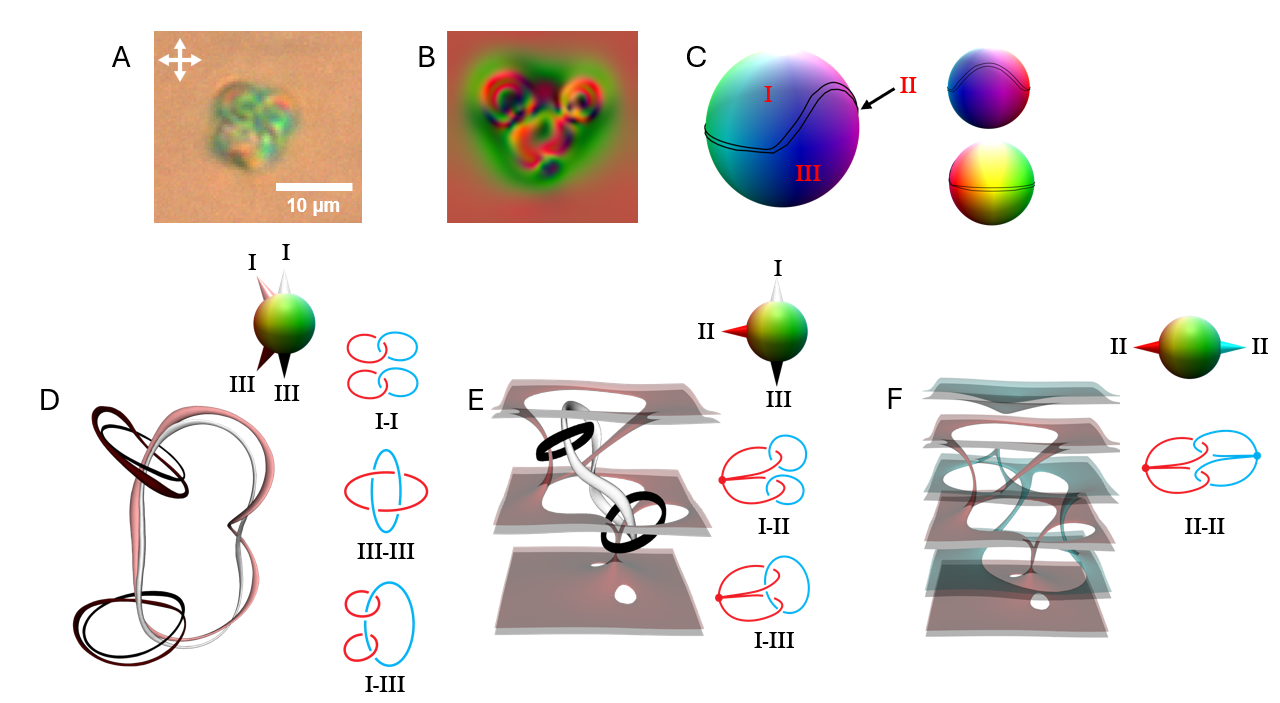}
\caption{An example of a $Q=2$ hopfion in a chiral nematic liquid crystal with diverse preimage linking structures. (A,B) Experimental (A) and computer-simulated (B) polarizing optical micrographs of the hopfion obtained between crossed polarizers with transmission axes parallel top the image edges (white double arrows). (C) Order parameter space (ground-state manifold) of the vectorized nematic director field where sub-spaces I, II, and III have different structures of preimages and their linking. (D) Ilustration of three types of preimage links that depend on the combinations of preimage pairing between regions I--III, exhibiting two Hopf links, or one solomon link, or a key-chain structure, all with the linking number equal two. (E,F) Three other types of preimage linkings also characterized by linking number of $Q=2$. For all preimages, the corresponding values of the constant order parameter are shown by cones of respective colors on the ground state manifold. Experimental samples and numerical parameters correspond to the planar chiral nematic cell comprising a mixture of pentylcyanobiphenyl nematic medium and chiral additive similar to the one described in \cite{TaSm19}.
}
\label{fig:heliknoton_motivation}
\end{figure}

In this paper, we model hopfions from a purely topological perspective by considering a class of smooth maps from $S^3$ to $S^2$ that naturally generalize the classical Hopf fibration.

Explicit analytic and numerical approximations to hopfions have been constructed in a number of earlier works within field-theoretic settings.  
In particular, Sutcliffe obtained families of knotted soliton solutions in the Skyrme-Faddeev model~\cite{Sut07}, 
where the field configurations are determined by minimizing an energy functional subject to a fixed Hopf index.  
Subsequently, Harland, Speight, and Sutcliffe proposed the so-called elastic rod approximation~\cite{HSS11}, 
which represents hopfions as geometrically deformed tubes whose centerlines follow closed space curves. 
These studies provide explicit realizations of hopfions as energetically stable configurations and offer intuitive geometric pictures of their shapes.  
Their approaches are formulated within specific field-theoretic models, most notably the Skyrme-Faddeev framework, and are primarily aimed at understanding the energetics and stability of such configurations.  

In contrast, the approach developed in the present paper is entirely model-independent and topological in nature.  
Instead of deriving approximate energy-minimizing fields, we regard hopfions as smooth maps $S^3 \to S^2$ equipped with prescribed singular structures.  
Specifically, we introduce the notion of a \emph{generalized Hopf map of order $n$}, formulated in terms of fold maps and their Stein factorizations,  
so that the linking of preimages becomes a central and intrinsic invariant.  
This formulation, rooted in the topological theory of singular fibers of differentiable maps~(\cite{GG73}, \cite{Sae04}),  
provides a unified framework that systematically classifies hopfions with arbitrary Hopf index and describes their geometric features independently of any field-theoretic assumptions.  
In this sense, our work complements the analytic constructions of \cite{Sut07} and \cite{HSS11} by shifting the focus from model-specific energetic stability to a purely topological classification of possible preimage linkings.  
Beyond providing a theoretical classification, this framework also enables a direct comparison between the topology of modeled hopfions and experimentally observed structures in physical systems.

Furthermore, we motivate and compare our theoretical model with experimental observations \cite{HTKNS} (Figure~\ref{fig:heliknoton_motivation}) of high-Hopf-index hopfions. 
We present imaging-based experimentally and numerically reconstructed preimages obtained from analyzing hopfion structures in the chiral liquid crystal material systems that correspond to distinct points under the Hopf map, demonstrating strong agreement between the mathematical model and experimental/computational results. 
Several specific examples from prior literature findings, like the composite topological solitons made of hopfions with opposite signs of Hopf index and ones with complex preimages are also discussed in the context of possible preimage linking configurations (see \cite{TAS18}, \cite{AcSm17PRX}).
This validation highlights the applicability of our approach in describing real-world topological structures. 
Moreover, we show that our main theorem provides a complete classification of the possible configurations of the fibers of two distinct points under the constructed high-Hopf-index hopfion model.

\subsection*{Acknowledgments}
Y.N. was supported by JSPS KAKENHI Grant Number JP23K12974.
Research of I.I.S. and D.H. was supported by the US Department of Energy, Office of Basic Energy Sciences, Division of Materials Sciences and Engineering, under contract DE-SC0019293 with the University of Colorado at Boulder.
Y.K. was supported by JSPS KAKENHI Grant Numbers JP23H05437, JP23K20791 and JP24K06744. 
Authors acknowledge the support and hospitality of the International Institute for Sustainability with Knotted Chiral Meta Matter (WPI-SKCM$^2$) at Hiroshima University, which helped to initiate this collaboration.
This work was supported by JSPS Program for Forming Japan's Peak Research Universities (J-PEAKS) Grant Number JPJS00420230011.
Finally, the authors are grateful to the anonymous referee for valuable suggestions.

\section{Preliminaries}

\subsection{Dehn twists and annulus twists}

To describe this construction explicitly, we identify the disk 
$D^2$ with the unit disk in $\CC$, 
\begin{equation}\label{eq: D2} 
D^2 = \{ r e^{i \theta} \mid 0 \leq r \leq 1,\, \theta \in \RR \} , 
\end{equation}
where its boundary is given by the unit circle 
\begin{equation}\label{eq: S1} 
S^1 = \partial D^2 = \{ e^{i \theta} \mid \theta \in \RR \}.
\end{equation}

Let $\Sigma$ be an oriented surface, and let $\gamma$ be a simple closed curve in $\Sigma$.  
A closed tubular neighborhood $N(\gamma)$ of $\gamma$ can be identified with the annulus  
$S^1 \times [0,1]$ via an orientation-preserving homeomorphism  
\begin{equation}\label{eq: phi_gamma}
\phi_\gamma\colon S^1 \times [0,1] \to N(\gamma).
\end{equation}

Here, we use the standard parameterization $(e^{i \theta }, t)$ for each point in $S^1 \times [0,1]$ and assume that $S^1 \times [0,1]$ inherits a natural orientation from the $(\theta, t)$-plane.  
We then define a map $T_\gamma\colon \Sigma \to \Sigma$ by  
\begin{equation}\label{eq: T_gamma}
T_{\gamma}(x) =
\begin{cases}
x, & x \in \Sigma \setminus N(\gamma), \\ 
(e^{i(\theta + 2\pi \phi(t))}, t), & (e^{i \theta}, t) \in N(\gamma),
\end{cases}
\end{equation}
where $\phi\colon \mathbb{R} \to \mathbb{R}$ is a smooth function satisfying  
\begin{equation}\label{eq: phi first}
\phi(t) =
\begin{cases}
0, & t \in (-\infty, 0), \\
1, & t \in (1, \infty),
\end{cases}
\end{equation}
and $\phi|_{[0,1]}$ is monotonically increasing.  
See Figure~\ref{fig:Dehn_twist}. 
\begin{figure}[htbp]
\centering\includegraphics[width=0.5\textwidth]{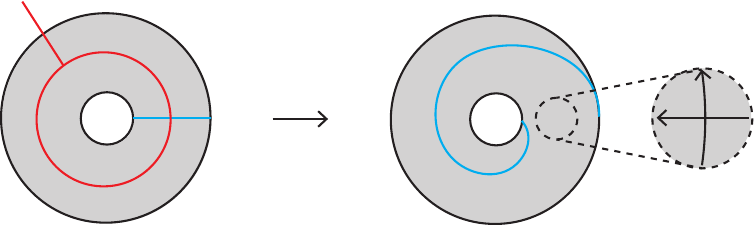}
\begin{picture}(400,0)(0,0)
\put(110,-1){$N(\gamma)$}
\put(220,-1){$N(\gamma)$}
\put(92,76){\color{red} $\gamma$}
\put(172,48){$T_\gamma$}
\put(288,60){$\theta$}
\put(272,40){$t$}
\end{picture}
\caption{A Dehn twist $T_\gamma\colon \Sigma \to \Sigma$.}
\label{fig:Dehn_twist}
\end{figure}
This map is called a (\emph{right-handed}) \emph{Dehn twist along $\gamma$}.  
Note that the Dehn twist $T_\gamma$ is defined only up to isotopy.  

Roughly speaking, an \emph{annulus twist}, which we now introduce, is obtained from a Dehn twist by taking its product with $[0,1]$.  
Let $M$ be an oriented $3$-manifold, and let $A$ be an annulus properly embedded in $M$, where the orientation of the core circle of $A$ is fixed.  
A closed tubular neighborhood $N(A)$ of $A$ can be identified with the product space $(S^1 \times [0,1]) \times [0,1]$ via an orientation-preserving homeomorphism  
\begin{equation}
(S^1 \times [0,1]) \times [0,1] \longrightarrow N(A)
\end{equation}
such that $(S^1 \times [0,1]) \times \{\frac{1}{2}\}$ is identified with the annulus $A$, and the orientation of $S^1$ agrees with the given orientation of the core of $A$.

Then, in analogy with the definition of Dehn twists, we define a map  
$\tau_A\colon M \to M$ by  
\begin{equation}\label{eq: tau_A}
\tau_A(x) \! = \!
\begin{cases}
x, & x \in M \setminus N(A), \\ 
((e^{i(\theta + 2\pi \phi(t))}, t), s), & ((e^{i \theta}, t), s) \in N(A) ,
\end{cases}
\end{equation}
where $\phi\colon \mathbb{R} \to \mathbb{R}$ is again a smooth function that satisfies
\begin{equation}\label{eq: phi second}
\phi(t) =
\begin{cases}
0, & t \in (-\infty, 0), \\
1, & t \in (1, \infty),
\end{cases}
\end{equation}
and $\phi|_{[0,1]}$ is monotonically increasing.  
See Figure~\ref{fig:annulus_twist}.  
\begin{figure}[htbp]
\centering\includegraphics[width=0.4\textwidth]{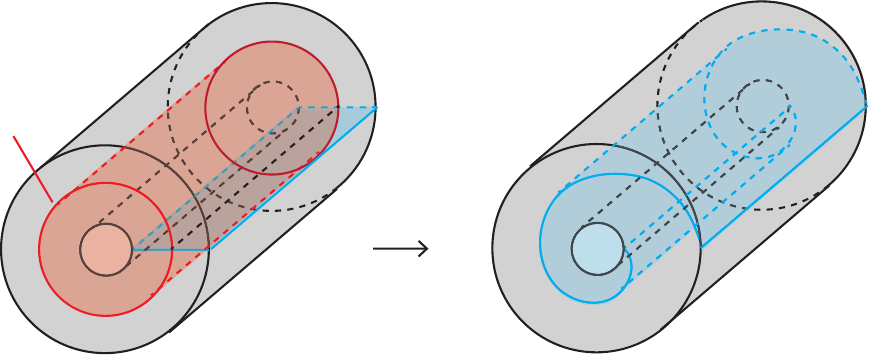}
\begin{picture}(400,0)(0,0)
\put(130,-1){$N(A)$}
\put(230,-1){$N(A)$}
\put(110,58){\color{red} $A$}
\put(188,37){$\tau_A$}
\end{picture}
\caption{An annulus twist $\tau\colon M \to M$.}
\label{fig:annulus_twist}
\end{figure}
This map is called a (\emph{right-handed}) \emph{annulus twist along $A$}.  
Like Dehn twists, an annulus twist can also be defined as a smooth map up to isotopy.  
Note that we have  
\begin{equation}\label{eq: tau|_partial M}
\tau_A|_{\partial M} = T_{a_1} \circ T_{a_2}^{-1},
\end{equation}
where $a_1, a_2 \subset \partial M$ are the boundary components of $A$.

\subsection{Fold maps and their generalization}
\label{subsec:Fold maps and their generalization}
Let $M$ be a closed, orientable $3$-manifold, and let $f$ be a smooth map from $M$ to an orientable surface $\Sigma$.  
A point $p \in M$ is called a \emph{singular point} of $f$ if the rank of the differential $df_p$ is less than $2$.  
We denote the set of singular points of $f$ by $S(f)$.

A point $p \in S(f)$ is called a \emph{fold point} if, in a neighborhood of $p$, the map $f$ is locally equivalent (up to diffeomorphism) to one of the following forms:
\begin{enumerate}[label=(\arabic*)]
\item
$(u, x, y) \mapsto (u, x^2 + y^2)$;  
\item
$(u, x, y) \mapsto (u, x^2 - y^2)$.
\end{enumerate}
In cases (1) and (2), $p$ is referred to as a \emph{definite fold point} and an \emph{indefinite fold point}, respectively.  
A smooth map $f \colon M \to \Sigma$ is called a \emph{fold map} if the set $S(f)$ consists only of fold points.

The notion of an indefinite fold point can be naturally generalized as follows.  
We identify $(x, y) \in \mathbb{R}^2$ with $x + i y \in \mathbb{C}$.  
For an integer $n \geq 2$, define a map $\omega_n \colon \RR^2 \to \RR$ by
\begin{equation}\label{eq: omega_n}
\omega_n(x, y) = \mathrm{Re}(x + i y)^n = r^n \cos(n\theta),
\end{equation}
where $\mathrm{Re}( \, \cdot \, )$ denotes the real part and $x + i y = r e^{i \theta}$.  
This map has an \emph{$n$-fold saddle point} at the origin $(0, 0)$.

A point $p \in S(f)$ is called an 
\emph{indefinite $n$-fold singularity}, or a \emph{multi-fold singularity}, if, 
in a neighborhood of $p$, the map $f$ 
is locally equivalent (up to diffeomorphism) to
\begin{equation}\label{eq:n-fold_singularity}
(u, x, y) \mapsto (u, \omega_n(x, y)).
\end{equation}
Note that when $n = 2$, this corresponds to an ordinary indefinite fold point.

Let $M$ be a closed, orientable $3$-manifold, and let $\Sigma$ be an orientable surface.  
Suppose $f \colon M \to \Sigma$ is a smooth map such that $S(f)$ consists only of indefinite multi-fold singularities.  
Then, $S(f)$ forms a link in $M$.  
We say that $f$ is \emph{simple} if the restriction $f|_{S(f)} \colon S(f) \to \Sigma$ is an embedding.  
We denote by $\mathcal{F}(M, \Sigma)$ the set of all smooth maps from $M$ to $\Sigma$ whose singular points consist only of simple indefinite multi-fold singularities.

\subsection{Stein Factorizations}

We begin with an informal description of the Stein factorization.
Given a smooth map $f\colon M \to \Sigma$, one considers the decomposition
of $M$ into connected components of the fibers of $f$.
The Stein factorization is obtained by collapsing each such connected
component to a single point.
Roughly speaking, the resulting quotient space provides a simplified description of the map that records how fibers are connected, such as how they split or merge, while suppressing their internal geometric details.

We now give a precise definition of this construction. 
Let $f\colon M \to \Sigma$ be a smooth map  
from a closed, orientable $3$-manifold to an orientable surface $\Sigma$.  
We say that two points $p_1, p_2 \in M$ are \emph{equivalent}  
if they satisfy $f (p_1) = f (p_2) =: q$ and  
belong to the same connected component of the preimage $f^{-1}(q)$.  
We denote by $W_f$ the quotient space of $M$  
with respect to this equivalence relation and by $q_f$ the quotient map.  
We define the map $\bar{f}\colon W_f \to \Sigma$ so that  
\begin{equation}\label{eq: f=barf circ qf}
f = \bar{f} \circ q_f.
\end{equation}
The composition $\bar{f} \circ q_f$, or the quotient space $W_f$, is referred to as the \emph{Stein factorization} of $f$. 

When $f \in \mathcal{F}(M,\Sigma)$, where the set $\mathcal{F}(M,\Sigma)$
is defined in Section~\ref{subsec:Fold maps and their generalization},
the quotient space $W_f$ has a particularly simple structure.
In this case, $W_f$ can be understood as a $2$-dimensional object
obtained by gluing together a finite collection of compact surfaces
along their boundary components.
We explain this description in more detail below.

First, let $q\in\Sigma$ be a regular value of $f$.
Then the preimage $f^{-1}(q)$ is a $1$-dimensional submanifold of $M$,
and hence a disjoint union of circles.
In the Stein factorization, each connected component of such a fiber is
collapsed to a single point.
Consequently, over the set of regular values, each point of $W_f$ has a
neighborhood that is homeomorphic to an open subset of the surface
$\Sigma$, and in particular $W_f$ is $2$-dimensional in a neighborhood
of such points.

Next, consider a neighborhood of a singular value of $f$.
By the definition of the class $\mathcal{F}(M,\Sigma)$, the map $f$ has only
indefinite multifold singularities.
Locally near such a singular point, $f$ is equivalent to the map in \eqref{eq:n-fold_singularity}.
In this local model, as one crosses the image of the singular set in the
target, connected components of a regular fiber split into or merge with other components.
After collapsing each connected component of the fiber, a neighborhood
in the quotient space $W_f$ is obtained by gluing $2$-dimensional pieces
along $1$-dimensional pieces. 
Therefore, $W_f$ is a $2$-dimensional polyhedron.
More precisely, it is obtained by gluing together a finite collection
of compact surfaces along their boundaries.
Here, by a \emph{$2$-dimensional polyhedron}, we mean the underlying space of a finite $2$-dimensional simplicial complex.

For example, in the simplest case where a fiber contains a single indefinite fold singularity, the Stein factorization of its saturated neighborhood (that is, a neighborhood that is a union of entire fibers) is a $2$-dimensional polyhedron that is homeomorphic to the product of a Y-shaped space and an interval.
Figure~\ref{fig:indefinite_fold_stein} illustrates this local picture,
showing a saturated neighborhood of the singular fiber, the corresponding
quotient space $W_f$, and the induced maps $q_f$ and $\bar f$. 
\begin{figure}[htbp]
\centering\includegraphics[width=0.8\textwidth]{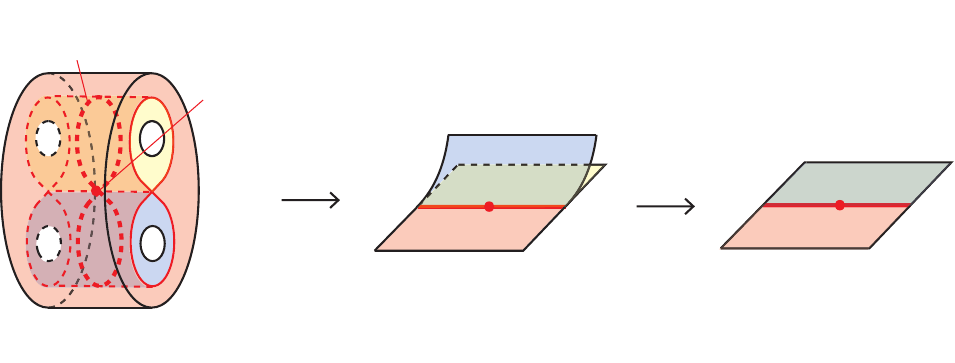}
\begin{picture}(400,0)(0,0)
\put(10,138){\color{red} A fiber containing a single} 
\put(10,125){\color{red} indefinite fold singularity}
\put(105,109){\color{red} An indefinite}
\put(105,96){\color{red} fold singularity}
\put(10,13){A saturated neighborhood}
\put(10,0){of the fiber}
\put(136,76){$q_f$}
\put(198,28){$W_f$}
\put(265,73){$\bar{f}$}
\put(330,28){$\Sigma$}
\end{picture}
\caption{The Stein factorization near a fiber containing a single indefinite fold singularity.}
\label{fig:indefinite_fold_stein}
\end{figure}

Moreover, with respect to such a simplicial structure, 
the polyhedron $W_f$ has no free faces.
Here, a \emph{free face} means a $1$-simplex that is contained in the
boundary of exactly one $2$-simplex.
The essential reason for the absence of free faces is that maps in
$\mathcal{F}(M,\Sigma)$ admit no definite fold points.
Indeed, a definite fold point corresponds, in the Stein factorization,
to the creation or annihilation of a connected component of a fiber,
which produces a $1$-simplex incident to only one $2$-simplex.
Since all fold singularities under consideration are indefinite, such
free faces do not occur in $W_f$. 

We refer the reader to \cite{Lev85}, \cite{Sae04}, and \cite{IK17} for
further details on Stein factorizations.
In particular, \cite[Chapter~1]{Lev85} contains a very detailed
description of the structure of Stein factorizations for smooth maps
from $3$-manifolds into $\mathbb{R}^2$.
In \cite[Chapter~1]{Lev85}, for \emph{stable maps}, whose singularities consist of
definite fold points, indefinite fold points (see
Section~\ref{subsec:Fold maps and their generalization}), and cusp points,
the structure of the Stein factorization over a saturated neighborhood
of a fiber is described according to the types of singular points contained in the
fiber.
Although the maps considered in the present paper are not stable in
general, the discussion of neighborhoods of fibers containing an
indefinite fold point extends to neighborhoods of fibers containing an
indefinite multifold singularity.
This extension is straightforward, since it essentially amounts to
increasing the number of branches.


\subsection{Hopf index}
\label{subsec:Hopf invariant}

Fix arbitrary orientations for $S^3$ and $S^2$.  
Let $\varphi\colon S^3 \to S^2$ be a smooth map.  
By Sard's theorem (see, e.g.,~\cite{GG73}), the image $\varphi(S(\varphi))$ of the set $S(\varphi)$ of critical points has Lebesgue measure zero.  
Thus, we can choose two points $q_1, q_2$ from $S^2 \setminus \varphi(S(\varphi))$.  
The preimage $\varphi^{-1} (q_i)$ of each $q_i$ ($i=1,2$) forms a link $L_i$ in $S^3$.  
Furthermore, from the pre-fixed orientations of $S^3$ and $S^2$, each component of $L_i$ naturally inherits an orientation.  

The \emph{Hopf index} (also known as the \emph{Hopf invariant}) is then defined as the linking number of $L_1$ and $L_2$.  
Here, the linking number $\Lk (L_1, L_2)$ is defined as follows.  

Let $n$ be the number of components of $L_2$.  
Then, the first homology group $H_1 (S^3 \setminus L_2)$ is a free abelian group $\mathbb{Z}^n$ of rank $n$, whose basis elements are uniquely determined by the orientations of $S^3$ and $L_2$.  
The homology class $[L_1]$ of $L_1$ in $H_1 (S^3 \setminus L_2)$ can be expressed as an element $(k_1, \dots, k_n) \in \mathbb{Z}^n$.  
The linking number $\Lk (L_1 , L_2)$ is then given by  
\begin{equation}\label{eq: Lk(L_1, L_2)}
\Lk (L_1 , L_2) = k_1 + \dots + k_n \in \mathbb{Z}.
\end{equation}
See \cite{Ste49} for more details.

\section{Mathematical model of high-Hopf-index hopfions}
\label{sec:Mathematical model of high-Hopf-index hopfions}

\subsection{Construction of generalized Hopf maps}
\label{subsec:Construction of generalized Hopf maps}
From this point onward, we construct a family of maps 
$\varphi_n\colon S^3 \to S^2$ that generalize the original Hopf fibration $\varphi\colon S^3 \to S^2$ 
for each $n \in \ZZ\setminus\{0\}$.  
This construction is achieved by carefully adapting the framework of stable maps developed in \cite{CT08, Sae96, IK17} to fit the specific setting of this paper.

For $n \in \NN$ and $k \in \{ 1, 2, \ldots, n \}$, 
let $x_{n, k} = \frac{1}{2} e^{2 \pi k i / n}$, 
$\varepsilon_1 = \frac{1}{4}$, and 
$\varepsilon_n = \frac{1}{2} \sin (\pi / n)$. 
Define  
\begin{equation}\label{eq: D_{n,k}}
D_{n, k} := \left\{ \frac{1}{2} e^{2 \pi k i / n}  + r e^{i\theta} 
\biggm| 0 \leq r \leq \varepsilon_n , \theta \in \mathbb{R} \right\}. 
\end{equation}
Then, $P_n := D^2 \setminus \Int \left(\bigcup_{k=1}^n D_{n, k} \right)$ is 
a disk with $n$ holes; see the left in Figure~\ref{fig:disk_with_holes}.  
\begin{figure}[htbp]
\centering\includegraphics[width=0.48\textwidth]{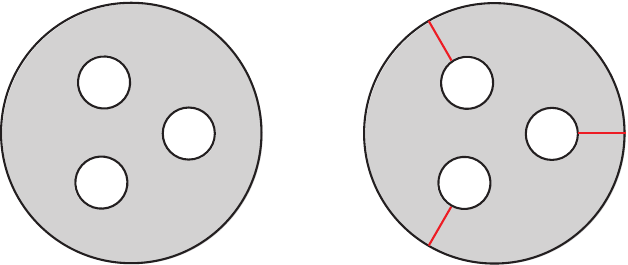}
\begin{picture}(400,0)(0,0)
\put(138,0){$P_3$}
\put(152,54){\tiny $D_{3,3}$}
\put(124,70){\tiny $D_{3,1}$}
\put(124,37){\tiny $D_{3,2}$}

\put(254,0){$P_3$}
\put(270,54){\tiny $D_{3,3}$}
\put(242,70){\tiny $D_{3,1}$}
\put(242,37){\tiny $D_{3,2}$}
\put(304,53){\color{red} $\alpha_{3,3}$}
\put(224,98){\color{red} $\alpha_{3,1}$}
\put(224,10){\color{red} $\alpha_{3,2}$}

\end{picture}
\caption{(Left) The disk $P_3$ with $3$ holes. 
(Right) The arcs $\alpha_{3,1} \alpha_{3,2}, \alpha_{3,3}$ in $P_3$.}
\label{fig:disk_with_holes}
\end{figure}

For each small disk $D_{n, k}$, we define a map $\eta_{n, k}\colon D_{n, k} \to D^2$ by  
\begin{equation}\label{eq: eta_{n,k}}
\eta_{n, k} ( x_{n, k} + r e^{i\theta} ) = \frac{1}{\varepsilon_n} r e^{(\theta - 2 \pi k / n) i}. 
\end{equation}
Let $\alpha_{n,1}, \dots, \alpha_{n,n}$ be arcs in $P_n$ as shown  
on the right side of Figure~\ref{fig:disk_with_holes}.  
Consider the product space $P_n \times S^1$.  
The products $A_{n, 1} := \alpha_{n,1} \times S^1, \dots,  
A_{n, n} := \alpha_{n,n} \times S^1$ are annuli properly embedded in  
$P_n \times S^1$.  
Define $\tau_n \colon P_n \times S^1 \to P_n \times S^1$ as the composition
\begin{equation}\label{eq: tau_n}
\tau_n = \tau_{A_{n, n}} \circ \cdots \circ \tau_{A_{n,1}}
\end{equation}
of the annulus twists $\tau_{A_{n,1}}, \ldots, \tau_{A_{n,n}}$ along the
annuli $A_{n,1}, \dots, A_{n,n}$.  
See Figure~\ref{fig:tau}, where $P_3 \times S^1$ is viewed as the result of identifying the top and bottom faces of $P_3 \times [0,1]$ depicted in the figure.  
\begin{figure}[htbp]
\centering\includegraphics[width=0.5\textwidth]{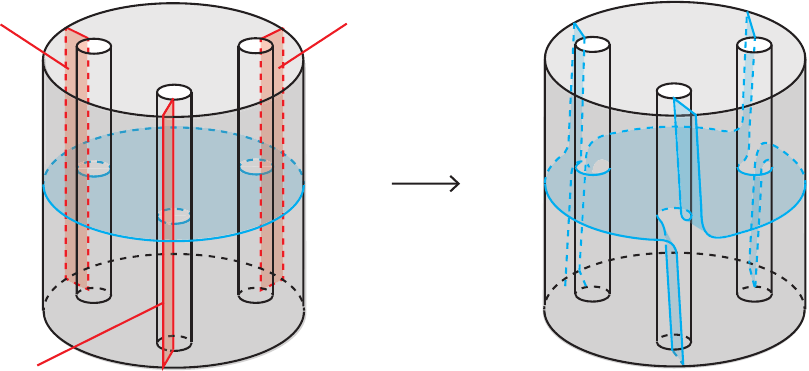}
\begin{picture}(400,0)(0,0)
\put(123,0){$P_3 \times S^1$}
\put(86,10){\color{red} $\alpha_{{3,1}}$}
\put(188,104){\color{red} $\alpha_{{3,2}}$}
\put(74,104){\color{red} $\alpha_{{3,3}}$}

\put(200,67){$\tau_3$}

\put(255,0){$P_3 \times S^1$}
\end{picture}
\caption{The map $\tau_3$. Here, $P_3 \times S^1$ is viewed as the result of identifying the top and bottom faces of $P_3 \times [0,1]$.}
\label{fig:tau}
\end{figure}
In Figures~\ref{fig:annulus_twist} and \ref{fig:tau}, the annuli along which the twists are performed
are highlighted (in red), while the auxiliary surface (in blue) is drawn
only as a ``visual guide'' to indicate how the twist propagates in the
ambient $3$-manifold.
More precisely, $\tau_n$ is supported in a small neighborhood of the
twisting annuli $A_{n,1}, \dots, A_{n,n}$ and is the identity outside this
neighborhood.
Inside a neighborhood of each annulus
$A_{n,i} = \alpha_{n,i} \times S^1$, the map rotates the $S^1$-direction
once while moving along the transverse direction of $A_{n,i}$.
Thus, one can think of $\tau_n$ as a composition of localized “shears”
supported near the annuli: objects disjoint from these neighborhoods are
unaffected, while objects intersecting them, for example the large boundary circle $\partial D \times \{\ast\}$ of the auxiliary surface
$P_n \times \{\ast\}$ shown in blue, are effectively twisted by a total
of $n$ turns.

Let $h_n\colon P_n \to \RR$ be the height function shown in
Figure~\ref{fig:height_function}.
\begin{figure}[htbp]
\centering\includegraphics[width=0.85\textwidth]{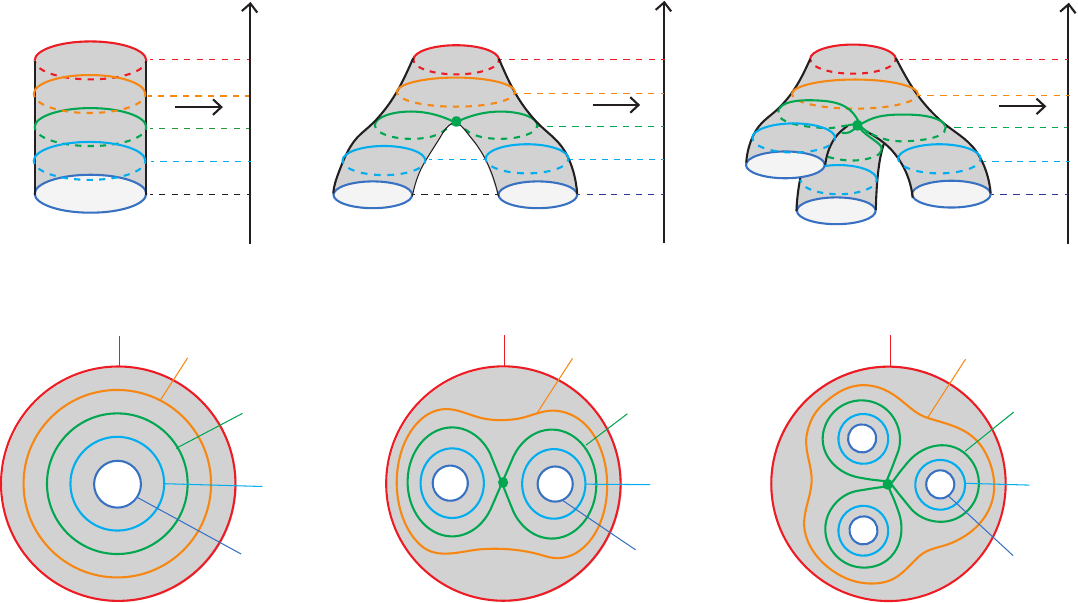}
\begin{picture}(400,0)(0,0)
\put(81,184){$h_1$}
\put(107,192){\color{red} $1$}
\put(106.5,170){\color{teal} $\frac{1}{2}$}
\put(107,146){\color{blue} $0$}
\put(43,130){$P_1$} 
\put(56,105){\color{red} $1$} 
\put(82,100){\color{orange} $\frac{3}{4}$} 
\put(104,76){\color{teal} $\frac{1}{2}$} 
\put(110,49){\color{cyan} $\frac{1}{4}$} 
\put(103,22){\color{blue} $0$} 
\put(52,0){$P_1$} 

\put(220,184){$h_2$}
\put(246,192){\color{red} $1$}
\put(245.5,170){\color{teal} $\frac{1}{2}$}
\put(246,146){\color{blue} $0$}
\put(167,130){$P_2$}
\put(186,105){\color{red} $1$} 
\put(212,100){\color{orange} $\frac{3}{4}$} 
\put(234,76){\color{teal} $\frac{1}{2}$} 
\put(240,49){\color{cyan} $\frac{1}{4}$} 
\put(235,24){\color{blue} $0$} 
\put(181,0){$P_2$} 

\put(357,184){$h_3$}
\put(383,192){\color{red} $1$}
\put(382.5,170){\color{teal} $\frac{1}{2}$}
\put(383,146){\color{blue} $0$}
\put(305,128){$P_3$} 
\put(316,105){\color{red} $1$} 
\put(344,100){\color{orange} $\frac{3}{4}$} 
\put(364,76){\color{teal} $\frac{1}{2}$} 
\put(369,49){\color{cyan} $\frac{1}{4}$} 
\put(363,22){\color{blue} $0$} 
\put(311,0){$P_3$} 
\end{picture}
\caption{The maps $h_n\colon P_n \to \RR$ ($n \in \NN$).}
\label{fig:height_function}
\end{figure}
For each $n$, the upper panel of the figure depicts the graph of $h_n$,
and the lower panel shows the corresponding level sets (contours) of
$h_n$ drawn on $P_n$.
These contour diagrams provide an intuitive picture of how the fibers
$h_n^{-1}(t)$ change as the level $t$ varies.

The map $h_1$ is non-singular, so every level set is a smooth
$1$-manifold.
In this case, each level set consists of a single circle, and its
topology does not change as $t$ crosses any value.
For $n \geq 2$, the map $h_n$ has a single critical point.
It is an $n$-fold saddle point in the following sense.
At the critical value, the level set develops an $n$-pronged
singularity.
Below the critical value, the level set consists of $n$ disjoint circles, while above the critical value these $n$ circles merge into a single circle.
Accordingly, as the level increases past the critical value, the level set undergoes a local transition in which $n$ circles merge into one, while if the level is decreased past the critical value, a single circle splits into $n$ circles.
Away from the critical value, the level sets are smooth, and the only topological change of the fibers occurs at this unique saddle.

We define the maps $g_0\colon [0,1] \times S^1 \to S^2 \subset \mathbb{R}^3$ and 
$g_\pm \colon D^2 \to S^2$ by 
\begin{align}
&g_0 (  t, e^{i\theta}  ) 
     = \left( \sin \left( \frac{3- 2t}{4} \pi \right) \cos \theta,\, \right. 
\sin \left( \frac{3- 2t}{4} \pi \right) \sin \theta,\, 
\left. \cos \left( \frac{3- 2t}{4} \pi \right) \right), \\
&g_+ (r e^{i\theta})  = 
\left( \sin \left( \frac{r \pi}{4} \right) \cos \theta,\, \right. 
\left. \sin \left( \frac{r \pi}{4} \right) \sin \theta,\, 
\cos \left( \frac{r \pi}{4} \right) \right), \\
& g_- (r e^{i\theta}) = 
\left( \sin \left( \left(1 - \frac{r}{4} \right) \pi \right) \cos \theta,\, \right.   
\sin \left( \left( 1 - \frac{r}{4} \right) \pi \right) \sin \theta,\,  \left. \cos \left( \left(1 - \frac{r}{4} \right) \pi \right) \right).  
\end{align}
See Figure~\ref{fig:annulus_to_sphere}.

\begin{figure}[htbp]
\centering\includegraphics[width=0.35\textwidth]{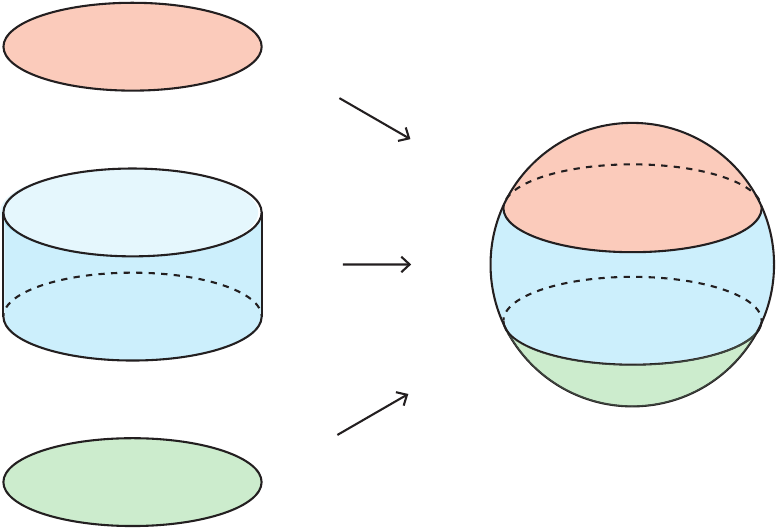}
\begin{picture}(400,0)(0,0)
\put(195,98){$g_+$}
\put(191,70){$g_0$}
\put(195,25){$g_-$}
\put(242,20){$S^2$}
\put(105,100){$D^2$}
\put(76,58){$[0,1] \times S^1$}
\put(105,18){$D^2$}
\end{picture}
\caption{The maps $g_0\colon  [0,1] \times S^1 \to S^2 \subset \RR^3$ and 
$g_\pm \colon D^2 \to S^2$.}
\label{fig:annulus_to_sphere}
\end{figure}

Now we are ready to define the map $\varphi_n\colon S^3 \to S^2$. 
We first fix the standard genus-1 Heegaard decomposition of $S^3$.
Explicitly, we identify $S^3$ with 
$\{ (z_1, z_2) \in \mathbb{C}^2 \mid |z_1|^2 + |z_2|^2 = 1 \}$, 
and we set
\begin{equation}\label{eq:Vpm}
V_+ = \{(z_1,z_2)\in S^3 \mid |z_1| \le |z_2| \},
\quad
V_- = \{(z_1,z_2)\in S^3 \mid |z_1| \ge |z_2| \}.
\end{equation}
Both $V_+$ and $V_-$ are solid tori, and they are glued together along their torus boundaries $\{(z_1,z_2)\in S^3 \mid |z_1| = |z_2| \}$. 
We identify each of these two subspaces $V_\pm$ with the standard solid torus
$D^2 \times S^1$.
Then we can write $S^3 = V_- \cup_{\phi} V_+$, where  
$\phi\colon \partial V_+ \cong \partial D^2 \times S^1 \to \partial D^2 \times S^1 \cong \partial V_-$ 
is given by  
\begin{equation}\label{eq: phi(etheta_1i, etheta_2i)}
\phi (e^{i \theta_1}, e^{i \theta_2}) = (e^{i \theta_2}, e^{i \theta_1}).
\end{equation}
We regard $P_n \times S^1$ and $D_{n, k} \times S^1$ ($k=1, \dots, n$)  
as naturally embedded subspaces of $V_- = D^2 \times S^1$, so that we have  
\begin{equation}\label{eq: V_-}
V_- = ( P_n \times S^1 ) \cup \left( \bigcup_{k=1}^n (D_{n, k} \times S^1) \right).
\end{equation}
Define the map $\varphi_{n, 0}\colon P_n \times S^1 \to S^2$ by  
\begin{equation}\label{eq: varphi_{n,0}}
\varphi_{n, 0} = g_0 \circ (h_n \times \operatorname{id}_{S^1}) \circ \tau_n^{-1}.
\end{equation}

Next, we define a map $f_n \colon D^2 \times S^1 \to D^2$ by
\begin{equation}\label{eq: f_n}
 f_n (r e^{i\theta_1}, e^{i \theta_2}) = r e^{(\theta_1 + n \theta_2)i}.
\end{equation}
Then, define the map $\varphi_{(n, k), -}\colon  
 D_{n, k} \times S^1 \to S^2$ by  
\begin{equation}\label{eq: varphi_{(n,k),-}}
\varphi_{(n, k), -} = g_- \circ f_1 \circ ( \eta_{n, k} \times \operatorname{id}_{S^1} ).
\end{equation}
This allows us to define a map  
\begin{equation}\label{eq: varphi_{n,-}}
\varphi_{n, -}\colon \bigsqcup_{k=1}^n (D_{n, k} \times S^1) \to S^2
\end{equation}
by imposing $\varphi_{n, -} (x) = \varphi_{(n, k), -} (x)$  
if $x \in D_{n, k} \times S^1$.  

Furthermore, define the map $\varphi_{n, +}\colon V_+ \to S^2$ by  
\begin{equation}\label{eq: varphi_{n,+}}
\varphi_{n, +} = g_+ \circ f_n.
\end{equation}
Since the map $\tau_n$ is defined up to smooth isotopy, we may assume that  
\begin{itemize}
\item 
for each point $x = (x_1, x_2, 1/\sqrt{2}) \in S^2$, the preimage  
$\varphi_{n, 0}^{-1} (x)$, which is a single circle on $\partial V_-$,  
coincides with $\phi ( \varphi_{n, +}^{-1} (x) )$;  
\item
for each point $x = (x_1, x_2, -\sqrt{2}) \in S^2$, the preimage  
$\varphi_{n, 0}^{-1} (x)$, which is a disjoint union of  
$n$ circles, coincides with $\varphi_{n, -}^{-1} (x)$.  
\end{itemize}

Now define the required map $\varphi_n\colon S^3 \to S^2$ by  
\begin{align}
\varphi_n (x) = 
\begin{cases}
\varphi_{n,+}(x), & x \in V_+, \\
\varphi_{n,0} (x), & x \in P_n \times S^1 \subset V_-, \\
\varphi_{n,-} (x), & x \in \bigsqcup_{k=1}^n (D_{n, k} \times S^1) \subset V_-.
\end{cases}
\label{eq:phi_n}
\end{align}
This map is well-defined due to the assumption stated above.  
Furthermore, by an isotopy of $\tau_n$, we may assume that $\varphi_n$ is smooth.  
We call this map the \textit{generalized Hopf map of order $n$}. 
Since the preimage of two distinct points in the upper hemisphere of $S^2$ under $\varphi_n$ forms a $(2,2n)$-torus link (see 
Example~\ref{example: preimages} below), 
we see that the Hopf index of $\varphi_n$ is $n$. 
Consider $S^3$ as the subset 
$ \{ (z_1, z_2) \in \CC^2 \mid |z_1|^2 + |z_2|^2 = 1 \} \subset \CC^2$. 
Define the involution $\rho\colon S^3 \to S^3$ by $\rho(z_1, z_2) = (z_1, \bar{z}_2)$, where $\bar{z}_2$ denotes the complex conjugate of $z_2$.
Then, for each $n \in \NN$, we define the smooth map $\varphi_{-n}\colon S^3 \to S^2$ by the composition $\varphi_{-n} = \varphi_n \circ \rho$.
It is straightforward to verify that the Hopf index of $\varphi_{-n}$ is $-n$.
In this manner, for any non-zero integer $n$, the generalized Hopf map $\varphi_n\colon S^3 \to S^2$ with Hopf index $n$ is well defined.

Although the earlier discussion of the map $f$ may have seemed technical, the following description of the singular set and the Stein factorization clarifies its natural underlying structure.

\begin{proposition}
\label{prop: singularity locus}
Let $n \in \ZZ\setminus\{0\}$.  
If $n=\pm 1$, the map $\varphi_n$ has no singular points, since the
function $h_n$ has no critical points.
If $|n|\ge 2$, the singular set $S(\varphi_n)$ of the generalized Hopf map $\varphi_n \colon S^3 \to S^2$ of order $n$ forms a trivial knot in $S^3$, and its image under $\varphi_n$ is the equator of $S^2$.  
Moreover, every point in $S(\varphi_n)$ is an indefinite $n$-fold singularity. 
In all cases, $\varphi_n \in \mathcal{F}(S^3, S^2)$.
\end{proposition}

\begin{proof}
This follows directly from the construction.  
In fact, the singularities of $\varphi_n$ arise solely from the $n$-fold saddle point of $h_n$, so $S(\varphi_n)$ consists entirely of indefinite $n$-fold singularities.  
Explicitly, we have  
\begin{equation}\label{eq: varphi_n}
S(\varphi_n) = \{ 0 \} \times S^1 \subset P_n \times S^1 \subset V_-,
\end{equation}
where we recall that $0 \in P_n \subset \CC$ is the center point of the domain $P_n$. 
Therefore, $S(\varphi_n)$ is exactly the core of the solid torus $V_-$, which is the trivial knot. 
\end{proof}

\begin{proposition}
\label{prop: stein factorization}
Let $n \in \ZZ\setminus\{0\}$.  
Consider the Stein factorization 
\begin{equation}\label{eq: Stein factorization of varphi_n}
S^3 \xrightarrow{q_{\varphi_n}} W_{\varphi_n} \xrightarrow{\overline{\varphi_n}} S^2
\end{equation}
of the generalized Hopf map of order $n$.  
Then, the space $W_{\varphi_n}$ can be described as the result of gluing the boundaries of $|n| + 1$ disks via homeomorphisms.  
Among these disks, one is mapped to the northern hemisphere of $S^2$, and the others to the southern hemisphere.  
See Figure~\ref{fig:stein_factorization}.
\end{proposition}
\begin{figure}[htbp]
\centering\includegraphics[width=0.45\textwidth]{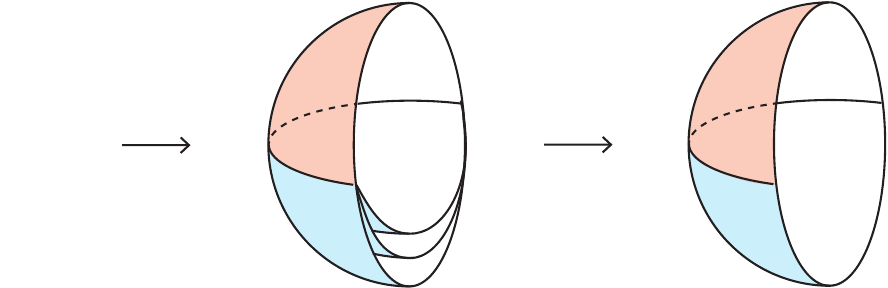}
\begin{picture}(400,0)(0,0)
\put(105,42){$S^3$}
\put(130,50){$q_{\varphi_3}$}
\put(222,50){$\overline{\varphi_3}$}
\put(176,0){$W_{\varphi_3}$}
\put(274,0){$S^2$}
\end{picture}
\caption{The Stein factorization $S^3 \xrightarrow{q_{\varphi_3}} W_{\varphi_3} \xrightarrow{\overline{\varphi_3}} S^2$. 
To clarify the structure of the quotient space $W_{\varphi_3}$ and the correspondence defined by $\overline{\varphi_3}$, both $W_{\varphi_3}$ and $S^2$ are depicted as halves.}
\label{fig:stein_factorization}
\end{figure}

\begin{proof}
By definition, the Stein factorizations of $\varphi_n$ and $\varphi_{-n}$ are identical.  
Hence, it suffices to prove the claim for $n \in \mathbb{N}$.
Decompose the 2-sphere $S^2$ into three regions:
\begin{equation}\label{eq: S_pm and S_0}
S_\pm := \{ (x_1, x_2, x_3) \in S^2 \mid \pm x_3 \geq 1/\sqrt{2} \}, \qquad
S_0 := \{ (x_1, x_2, x_3) \in S^2 \mid |x_3| \leq 1/\sqrt{2} \}.
\end{equation}
Note that 
$\Im \varphi_{n,\pm} = \Im g_\pm = S_\pm$ and $\Im \varphi_{n,0} = \Im g_0 = S_0$.

Let $x \in S_+$. 
Then $g_+^{-1}(x)$ is a single point, and its preimage under $f_n$ is a single circle.  
Thus, the preimage of $x$ under $\varphi_{n,+}$ is a single circle.  
The map $f_n$ coincides with the quotient map 
\begin{equation}\label{eq: q_varphi_{n,+}}
q_{\varphi_{n,+}} \colon V_+ \to W_{\varphi_{n,+}}, 
\end{equation}
where $W_{\varphi_{n,+}}$ is a disk mapped to $S_+$ by $\overline{\varphi_{n,+}}$.

Now let $x \in S_-$. 
Then $g_-^{-1}(x)$ is a single point, and under each map $\varphi_{(n,k),-}$, its preimage is a single circle.  
Thus, $\varphi_{n,-}^{-1}(x)$ consists of $n$ circles.  
By the definition of $\varphi_{(n,k),-}$, the quotient map
\begin{equation}\label{eq: q_{varphi_{n}}}
q_{\varphi_{n,-}} \colon \bigsqcup_{k=1}^n (D_{n,k} \times S^1) \to W_{\varphi_{n,-}}
\end{equation}
sends each solid torus $D_{n,k} \times S^1$ to a disk, thus, $W_{\varphi_{n,-}}$ consists of $n$ disks, which is mapped to $S_-$ by $\overline{\varphi_{n,-}}$.

Finally, let $x = (x_1, x_2, x_3) \in S_0$. 
Then $g_0^{-1}(x)$ is a single point.  
If $x_3 \geq 0$, its preimage under $h_n \times \mathrm{id}_{S^1}$ is a single circle, and hence $\varphi_{n,0}^{-1}(x)$ is a single circle.  
If $x_3 < 0$, the preimage under $h_n \times \mathrm{id}_{S^1}$ consists of $n$ circles, so $\varphi_{n,0}^{-1}(x)$ consists of $n$ circles.
Therefore, quotient space $W_{\varphi_{n,0}}$, given by
\begin{equation}\label{eq: q_{varphi_{n,o}}}
q_{\varphi_{n,0}} \colon (P_n \times S^1) \to W_{\varphi_{n,0}},
\end{equation}
is homeomorphic to $K_{1,n+1} \times S^1$, where $K_{1,n+1}$ denotes the star graph with $n+1$ leaves, or equivalently, the cone over $n+1$ points.  
The map $\overline{\varphi_{n,0}}$ sends the subset $\{v\} \times S^1$ to the equator of $S^2$, where $v$ is the central vertex of degree $n+1$.

Since $W_{\varphi_n}$ is constructed by gluing 
$W_{\varphi_{n,+}}$, $W_{\varphi_{n,-}}$, and $W_{\varphi_{n,0}}$  
along their boundaries via homeomorphisms, the claim follows.
\end{proof}

\subsection{Classification of preimages of two distinct points}

\begin{example}
\label{example: preimages}
\begin{enumerate}[label=(\arabic*)]
\item 
For $n = 1$, it is straightforward to see that the generalized Hopf map  
$\varphi_1\colon S^3 \to S^2$ of order $1$  
is nothing but the Hopf map $\varphi$. 
Thus, the preimage of any two distinct points of $S^2$ forms the Hopf link. 
\item
Let $n = 2$.  
Then the generalized Hopf map  
$\varphi_2\colon S^3 \to S^2$ of order $2$ is a simple fold map whose set of singular points  
$S (\varphi_2)$ forms the trivial knot consisting only of indefinite fold points.  
The image $\varphi_2 (S (\varphi_2))$ of the singular points is the equator of $S^2$.  
Figure~\ref{fig:preimage_2}~(i)--(vi) 
illustrate the preimages of two points in $S^2$ under the map $\varphi_2$. 
\begin{figure}[htbp]
\centering\includegraphics[width=0.95\textwidth]{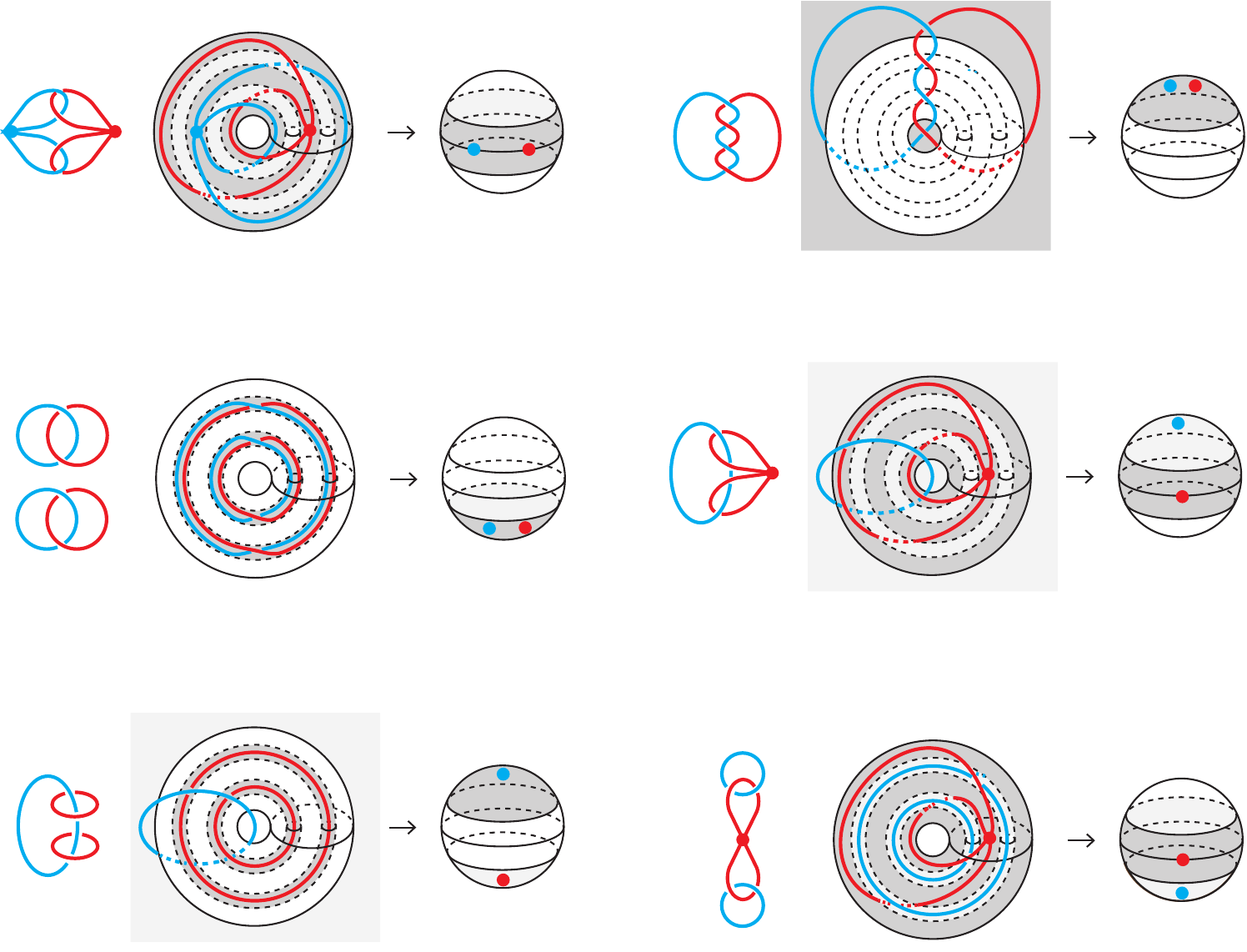}
\begin{picture}(400,0)(0,0)
\put(72,226){(i)}
\put(290,226){(ii)}
\put(71,115){(iii)}
\put(290,115){(iv)}
\put(71,0){(v)}
\put(290,0){(vi)}
\end{picture}
\caption{The preimage under $\varphi_2$ of two distinct points on $S^2$ in the following configurations:  
(i) both on the equator;  
(ii) both in the upper hemisphere;  
(iii) both in the lower hemisphere;  
(iv) one on the equator and one in the upper hemisphere;  
(v) one in the upper hemisphere and one in the lower hemisphere;  
(vi) one on the equator and one in the lower hemisphere.}
\label{fig:preimage_2}
\end{figure}

Since $\varphi_2$ is a fold map, any map that is 
sufficiently close to $\phi_2$ in the space $C^\infty (S^3, S^2)$ of smooth maps, equipped with the Whitney $C^\infty$-topology, has the 
same configurations of 
preimages of two distinct points in $S^2$.

\item
By the same arguments as above, we see that all 
possible configurations of 
preimages of two distinct points in $S^2$ under the generalized Hopf map 
$\varphi_3$ of order $3$ are as illustrated in Figure~\ref{fig:preimage_3}. 
\begin{figure}[htbp]
\centering\includegraphics[width=0.85\textwidth]{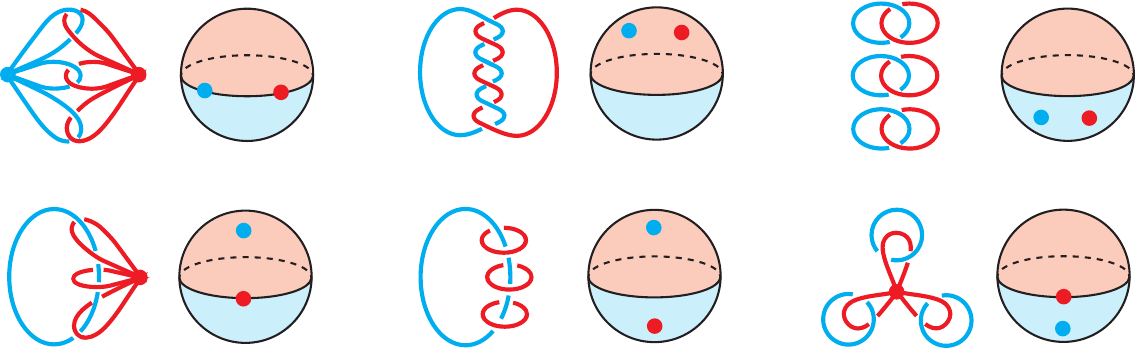}
\begin{picture}(400,0)(0,0)
\end{picture}
\caption{All possible patterns of spatial graphs obtained as preimages of two distinct points in $S^2$ under the generalized Hopf map $\varphi_3$ of order $3$.}
\label{fig:preimage_3}
\end{figure}
\end{enumerate}
\end{example}

From the above construction and Example~\ref{example: preimages}, we obtain the following. 
\begin{theorem}
Let $n$ be a natural number. 
The preimage of any two distinct points in $S^2$ under the generalized Hopf map $\varphi_n$ of order $n$ is one of the spatial graphs \textup{(i)--(vi)} illustrated in Figure~\ref{fig:preimage_n}. 
In the case of the map $\varphi_{-n}$, the corresponding preimages are the mirror images of those shown above.
\begin{figure}[htbp]
\vspace{1em}
\centering\includegraphics[width=1.0\textwidth]{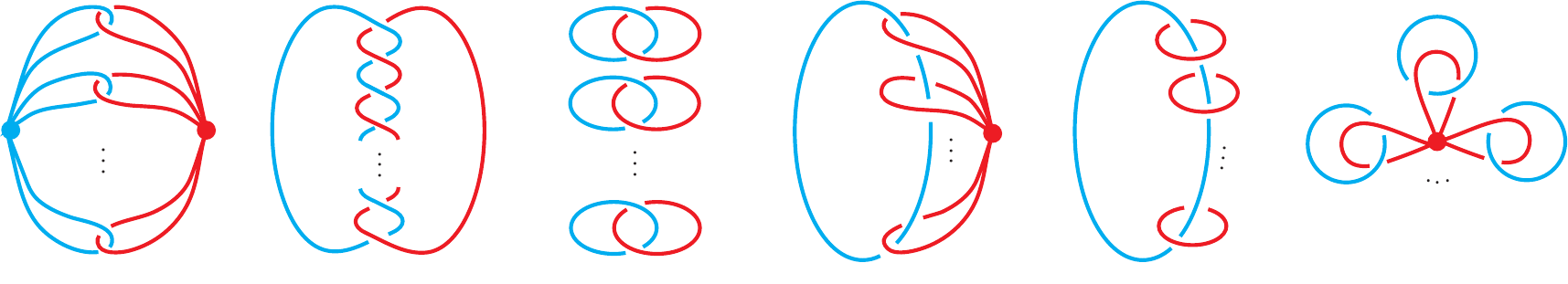}
\begin{picture}(400,0)(0,0)
\put(9,0){(i)}
\put(10,90){\footnotesize $1$}
\put(10,71){\footnotesize $2$}
\put(10,32){\footnotesize $n$}

\put(82,0){(ii)}
\put(65,93){\footnotesize $n$ full twists}

\put(145,0){ (iii)}
\put(133,77){\footnotesize $1$}
\put(133,58){\footnotesize $2$}
\put(133,24){\footnotesize $n$}

\put(220,0){(iv)}
\put(218,78){\footnotesize $1$}
\put(218,62){\footnotesize $2$}
\put(218,23){\footnotesize $n$}

\put(292,0){\footnotesize (v)}
\put(293,76){\footnotesize $1$}
\put(296,61){\footnotesize $2$}
\put(292,26){\footnotesize $n$}

\put(368,0){(vi)}
\put(374,89){\footnotesize $1$}
\put(340,62){\footnotesize $2$}
\put(410,62){\footnotesize $n$}

\end{picture}
\caption{The six distinct patterns of spatial graphs that arise as preimages of two distinct points in $S^2$ under the generalized Hopf map $\varphi_n$ of order $n$. 
In the case of $n=1$, the spatial graphs (i)--(vi) are all Hopf links. 
This summary of distinct patterns is consistent with experimental results shown in Figure~\ref{fig:heliknoton_motivation} and generalizes such structures to arbitrarily high Hopf index values.}
\label{fig:preimage_n}
\end{figure}
\end{theorem}

\section{Several variations}

\subsection{Concentric hybridization of generalized Hopf maps}

Beyond individual hopfions, more complex field configurations, such as those discussed in~\cite{AcSm17PRX}, can be modeled by combining generalized Hopf maps.  
To construct such configurations, we define a continuous map  
\begin{equation}\label{eq: varphi_mn}
\varphi_{m,n} \colon S^3 \to S^2
\end{equation}
for $m, n \in \ZZ\setminus\{0\}$, by merging two given maps  
$\varphi_m \colon S^3 \to S^2$ and $\varphi_n \colon S^3 \to S^2$.

We view $S^3$ as the unit sphere in $\mathbb{R}^4$, that is,
\begin{equation}\label{eq: S3}
S^3 = \left\{ (x, y, z, w) \in \mathbb{R}^4 \mid x^2 + y^2 + z^2 + w^2 = 1 \right\}.
\end{equation}
To facilitate the construction, we consider the equatorial $2$-sphere
\begin{equation}\label{eq: S2_0}
S^2_0 = \left\{ (x, y, z, 0) \in S^3 \right\} \subset S^3.
\end{equation}
We collapse $S^2_0$ to a point, obtaining the quotient space
$S^3 / S^2_0$, 
which is homeomorphic to the \emph{wedge sum} $S^3 \vee S^3$: two copies of $S^3$ joined at a single point (the image of $S^2_0$ under the quotient map). 
Using this identification,
$S^3 / S^2_0 \cong S^3 \vee S^3$,
we define a gluing scheme that allows us to construct the map $\varphi_{m,n}$ in a consistent and well-defined way.

As in Section~3.1, we represent $S^3$ as the union
$S^3 = V_+ \cup_\varphi V_-$, where $V_+ = V_- = D^2 \times S^1$.  
We choose the identification so that a point on the core circle $\{0\} \times S^1 \subset V_+$ is mapped to the wedge point of $S^3 \vee S^3$.  
This choice is compatible with all generalized Hopf maps $\varphi_n$, as they send the core circle to the north pole $(0,0,1) \in S^2$.  
Hence, the maps $\varphi_m$ and $\varphi_n$ agree at the point where the two copies of $S^3$ are glued together.

Under this setup, we define the \emph{concentric hybridization} $\varphi_{m,n} \colon S^3 \to S^2$ 
as the composition:
\begin{equation}\label{eq: varphi_{m,n}}
S^3 \xrightarrow{q} S^3 / S^2_0 \cong S^3 \vee S^3 \xrightarrow{\varphi_m \vee \varphi_n} S^2,
\end{equation}
where $q$ is the quotient map collapsing $S^2_0$ to a point, and
\begin{equation}\label{eq: varphi_m vee varphi_n}
\varphi_m \vee \varphi_n \colon S^3 \vee S^3 \to S^2
\end{equation}
is defined piecewise: it restricts to $\varphi_m$ on the first copy of $S^3$, and to $\varphi_n$ on the second.
This construction realizes $\varphi_{m,n}$ as a \emph{homotopy-theoretic concatenation} of the maps $\varphi_m$ and $\varphi_n$.  
In particular, $\varphi_{m,n}$ represents the sum of homotopy classes:
\begin{equation}\label{eq: [varphi_{m,n}]}
 [\varphi_{m,n}] = [\varphi_m] + [\varphi_n] \in \pi_3(S^2) \cong \mathbb{Z}.
\end{equation}

As a concrete example, Figure~\ref{fig:preimage_hybrid} illustrates the preimages of two distinct points in $S^2$ under the map $\varphi_{1,-1}$ \cite{TAS18}.
The two Hopf links shown represent the preimages under $\varphi_1$ and $\varphi_{-1}$, respectively, and correspond to positively and negatively oriented Hopf links as described in Section~\ref{subsec:Hopf invariant}.  
The resulting configuration corresponds to the one depicted in \cite[Figures~7 and 8]{AcSm17PRX}, representing a topological soliton of Hopf index $Q = 0$, formed by concentric hybridization of components with $Q = 1$ and $Q = -1$.
\begin{figure}[htbp]
\centering
\includegraphics[width=0.65\textwidth]{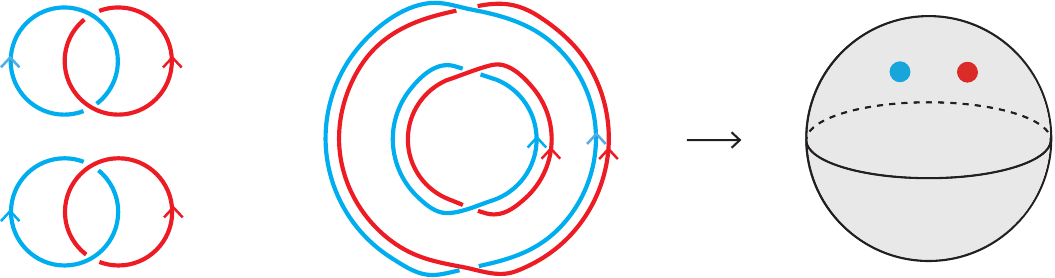}
\begin{picture}(400,0)(0,0)
\put(124,45){$\cong$}
\put(236,55){$\varphi_{1,-1}$}
\end{picture}
\caption{Preimages of two distinct points in $S^2$ under  
the concentric hybridization $\varphi_{1,-1}$, combining $\varphi_1$ and $\varphi_{-1}$.}
\label{fig:preimage_hybrid}
\end{figure}

\subsection{Solitons with broken axial symmetry}

The mathematical model $\varphi_n$ of the high-Hopf index hopfion introduced in Section~\ref{sec:Mathematical model of high-Hopf-index hopfions} is deliberately formulated in a standard and idealized framework, in order to highlight its essential topological structure.
Even minor modifications to a map in $\mathcal{F}(S^3, S^2)$---where $\mathcal{F}(S^3, S^2)$ denotes the set of all smooth maps from $M$ to $\Sigma$ whose singularities consist solely of simple indefinite multi-fold singular points---can lead to qualitatively different preimages.

In this subsection, we present a modified version of the generalized Hopf map $\varphi_2$ of order $2$, which also belongs to $\mathcal{F}(S^3, S^2)$ and is adapted to reflect the experimental findings reported in~\cite{TaSm18}.
The hopfion discussed in \cite{TaSm18} has Hopf index $Q = -2$, but for simplicity of exposition, we construct its mirror image, namely the one with Hopf index $Q = 2$.

First, we embed $S^1 \times [0,1]$ into $D^2$ via the map $\eta_- \colon S^1 \times [0,1] \to D^2$, defined by $\eta_-(e^{i\theta}, t) = (1 - \frac{t}{2}) e^{i\theta}$.  
We also define an embedding of $D^2$ into itself by $\eta_+ \colon D^2 \to D^2$, given by $\eta_+(r e^{i\theta}) = \frac{r}{2} e^{i\theta}$.
We continue to use the same notation as in Section~\ref{subsec:Construction of generalized Hopf maps}.
Recall that the 3-sphere $S^3$ can be decomposed into two solid tori, denoted by $V_+$ and $V_-$.  
We construct a variation of $\varphi_2$ by modifying only the map $\varphi_{2,+} \colon V_+ \to S^2$, replacing it with a new map $\hat{\varphi}_{2,+} \colon V_+ \to S^2$ as described below.

\begin{figure}[htbp]
\centering
\includegraphics[width=0.8\textwidth]{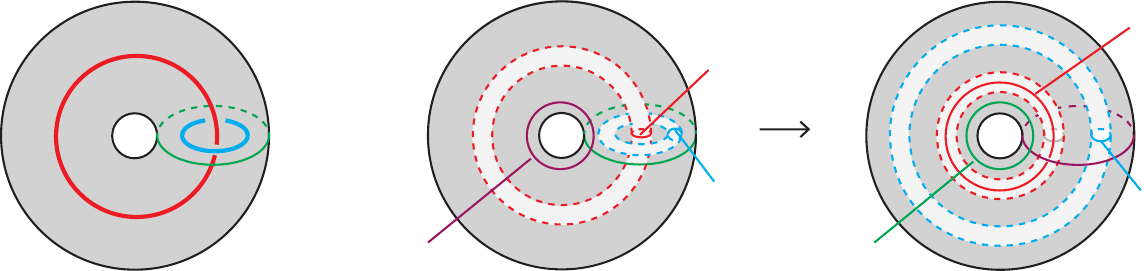}
\begin{picture}(400,0)(0,0)
\put(55,0){$L \subset V_+$}
\put(165,0){$V_+ \setminus \Int (N(L))$}
\put(310,0){$P_2 \times S^1$}
\put(241,52){\color{teal} $\mu$}
\put(148,15){\color{violet} $\lambda$}
\put(243,70){\color{red} $\mu_1$}
\put(244,34){\color{blue} $\mu_2$}
\put(266,15){\color{teal} $\psi(\mu)$}
\put(372,52){\color{violet} $\psi(\lambda)$}
\put(368,84){\color{red} $\psi(\mu_1)$}
\put(369,28){\color{blue} $\psi(\mu_2)$}
\put(260,60){$\psi$}
\put(260,43){$\cong$}
\end{picture}
\caption{(Left) The link $L$ in the solid torus $V_+$. (Right) The diffeomorphism $\psi$ from the exterior $V_+ \setminus \Int (N(L))$ of the link $L$ in $V_+$ to the product space $P_2 \times S^1$.}
\label{fig:solid_torus_hopf_link}
\end{figure}

Let $L$ denote the link inside the solid torus $V_+$ depicted on the left in Figure~\ref{fig:solid_torus_hopf_link}.  
The space $V_+ \setminus \Int(N(L))$ is diffeomorphic to $P_2 \times S^1$, where $\Int( \, \cdot \,)$ denotes the interior.  
Let $\psi \colon V_+ \setminus \Int(N(L)) \to P_2 \times S^1$ be such a diffeomorphism.  
The right-hand side of Figure~\ref{fig:solid_torus_hopf_link} shows the images of the meridional curve $\mu$ and longitudinal curve $\lambda$ of $V_+$, as well as the meridional curves $\mu_1$ and $\mu_2$ of the two components of the link $L$, under the map $\psi$.

We now define a map $\varphi_{2,+}' \colon V_+ \setminus \Int(N(L)) \to S^2$ by 
\begin{equation}\label{eq: varphi_{2,+}'}
\varphi_{2,+}' = g_+ \circ \eta_- \circ (h_2 \times \id_{S^1}) \circ \psi , 
\end{equation}
where we recall that the maps $h_n$ and $g_+$ are described in Figures~\ref{fig:height_function} and \ref{fig:annulus_to_sphere}, respectively.
We identify each component of the tubular neighborhood $N(L)$ with a solid torus $D^2 \times S^1$ such that the blackboard framing induced by Figure~\ref{fig:solid_torus_hopf_link} aligns with the longitude $\{\ast\} \times S^1$, where $\ast$ is a point on $\partial D^2$.  
Then, define the map  
\begin{equation}\label{eq: varphi_{2,+}''}
\varphi_{2,+}'' \colon N(L) \cong (D^2 \times S^1) \sqcup (D^2 \times S^1) \to S^2  
\end{equation}
by  
\begin{equation}\label{eq: definition of varphi_{2,+}''}
\varphi_{2,+}'' = g_+ \circ \eta_+ \circ (f_1 \sqcup f_1).
\end{equation}
As in the construction of the generalized Hopf maps, we may assume the following:
\begin{itemize}
\item  
For each point $x \in \Im(\varphi_{2,+}') \cap \Im(\varphi_{2,+}'')$, the preimages coincide:
$(\varphi_{2,+}')^{-1}(x) = (\varphi_{2,+}'')^{-1}(x)$.  
\item  
For each point $x = (x_1, x_2, 1/\sqrt{2}) \in S^2$, the preimage $\varphi_{n,0}^{-1}(x)$ coincides with $\phi((\varphi_{2,+}'')^{-1}(x))$.
\end{itemize}

Finally, we define the modified map $\hat{\varphi}_2 \colon S^3 \to S^2$ by  
\begin{equation}\label{eq: hatvarphi_2(x)}
\hat{\varphi}_2(x) = 
\begin{cases}
\varphi_{2,+}''(x), &  x \in N(L), \\
\varphi_{2,+}'(x), & x \in V_+ \setminus \Int(N(L)), \\
\varphi_2(x), &  x \in V_-.
\end{cases}
\end{equation}
This map is well-defined by construction.  
Moreover, using an isotopy of $\tau_2$, we may further assume that $\hat{\varphi}_2$ is smooth.

From the construction, we can verify the following properties of $\hat{\varphi}_2$ by an argument analogous to that in Propositions~\ref{prop: singularity locus} and \ref{prop: stein factorization}: 
\begin{enumerate}
\item 
The singular set $S(\hat{\varphi}_2)$ of the map $\hat{\varphi}_2 \colon S^3 \to S^2$ is a $2$-component trivial link in $S^3$. Its image in $S^2$ under $\hat{\varphi}_2$ is the (disjoint) union of the equator and the boundary of the arctic circle. 
These circles divide the codomain $S^2$ into three regions: the arctic circle (a disk), the remainder of the northern hemisphere (an annulus), and the southern hemisphere. 
Note that this decomposition of $S^2$ does not correspond to the decomposition of the domain $S^3$ used in the definition of $\hat{\varphi}_2$. 
\item
Every point in $S(\hat{\varphi}_2)$ is an indefinite $2$-fold singularity. 
Hence, $\hat{\varphi}_2 \in \mathcal{F}(S^3, S^2)$.
\item 
Consider the Stein factorization 
\begin{equation}\label{eq: Stein factorization of hatvarphi_2}
S^3 \xrightarrow{q_{\hat{\varphi}_2}} W_{\hat{\varphi}_2} \xrightarrow{\overline{\hat{\varphi}_2}} S^2. 
\end{equation}
The quotient space $W_{\hat{\varphi}_2}$ consists of five regions: four disks and one annulus. 
See Figure~\ref{fig:stein_factorization2}.
\begin{figure}[htbp]
\centering\includegraphics[width=0.45\textwidth]{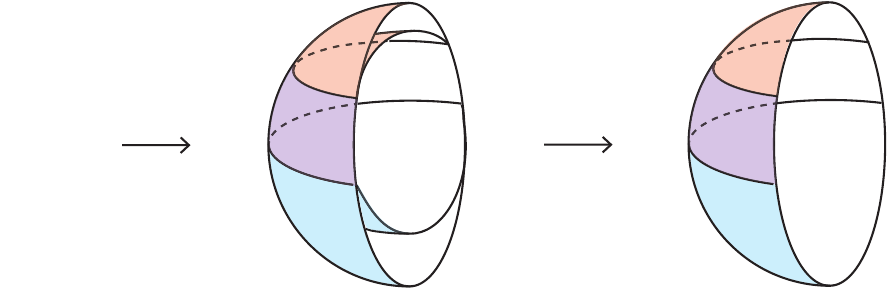}
\begin{picture}(400,0)(0,0)
\put(105,42){$S^3$}
\put(130,50){$q_{\hat{\varphi}_2}$}
\put(222,50){$\overline{\hat{\varphi}_2}$}
\put(176,0){$W_{\hat{\varphi}_2}$}
\put(274,0){$S^2$}
\end{picture}
\caption{The Stein factorization $S^3 \xrightarrow{q_{\hat{\varphi}_2}} W_{\hat{\varphi}_2} \xrightarrow{\overline{\hat{\varphi}_2}} S^2$. 
For clarity, both $W_{\hat{\varphi}_2}$ and $S^2$ are shown as halves.}
\label{fig:stein_factorization2}
\end{figure}
\end{enumerate}

Figure~\ref{fig:preimage_hat_varphi_2_1} illustrates the preimages of representative points from each of these regions, as well as the shared boundaries between them.
A key difference between the maps $\varphi_2$ and $\hat{\varphi}_2$ is that, under $\hat{\varphi}_2$, the preimage of a point near the north pole already forms a Hopf link, whereas under $\varphi_2$, the preimages are always trivial knots or links.
\begin{figure}[htbp]
\centering\includegraphics[width=0.85\textwidth]{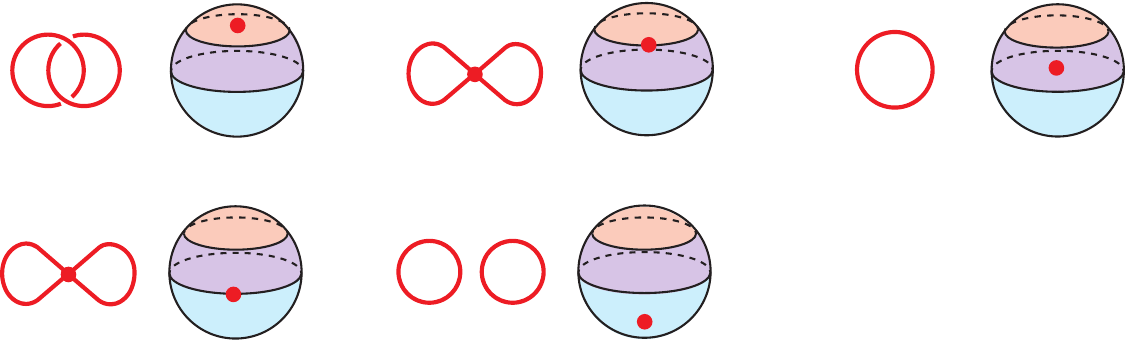}
\caption{Preimages of typical points from each region of $S^2$ under $\hat{\varphi}_2$, including their shared boundaries.}
\label{fig:preimage_hat_varphi_2_1}
\end{figure}

Figure~\ref{fig:preimage_hat_varphi_2_2} presents a classification of all possible spatial graph patterns that appear as preimages of two distinct points in $S^2$ under the map $\hat{\varphi}_2$.
Note that the six preimages shown at the bottom of Figure~\ref{fig:preimage_hat_varphi_2_2} are identical to those in Figure~\ref{fig:preimage_2} since the map $\hat{\varphi}_2$ differs from $\varphi_2$ only on the part of the domain mapped to the arctic circle.

We have described above a generalization only of $\varphi_2$ to match the experimental result reported in \cite{TaSm18}, but naturally, a similar type of modification can be considered for any $\varphi_n$.  
An interesting open question is whether, for a given knot or link $L$ and a non-zero integer $n \in \ZZ\setminus\{0\}$,  
there exists a map $\varphi \in \mathcal{F}(S^3, S^2)$ 
of Hopf index $n$ such that the preimage of a regular value in $S^2$ is isotopic to $L$.  
An ultimate approach to this problem would be to classify the set $\mathcal{F}(S^3, S^2)$, and to describe the topology of the preimage of a regular value for each such map.  
Although this goal may seem ambitious, it is not entirely out of reach, and we leave it as a subject for future investigation.

\begin{figure}[htbp]
\centering\includegraphics[width=0.85\textwidth]{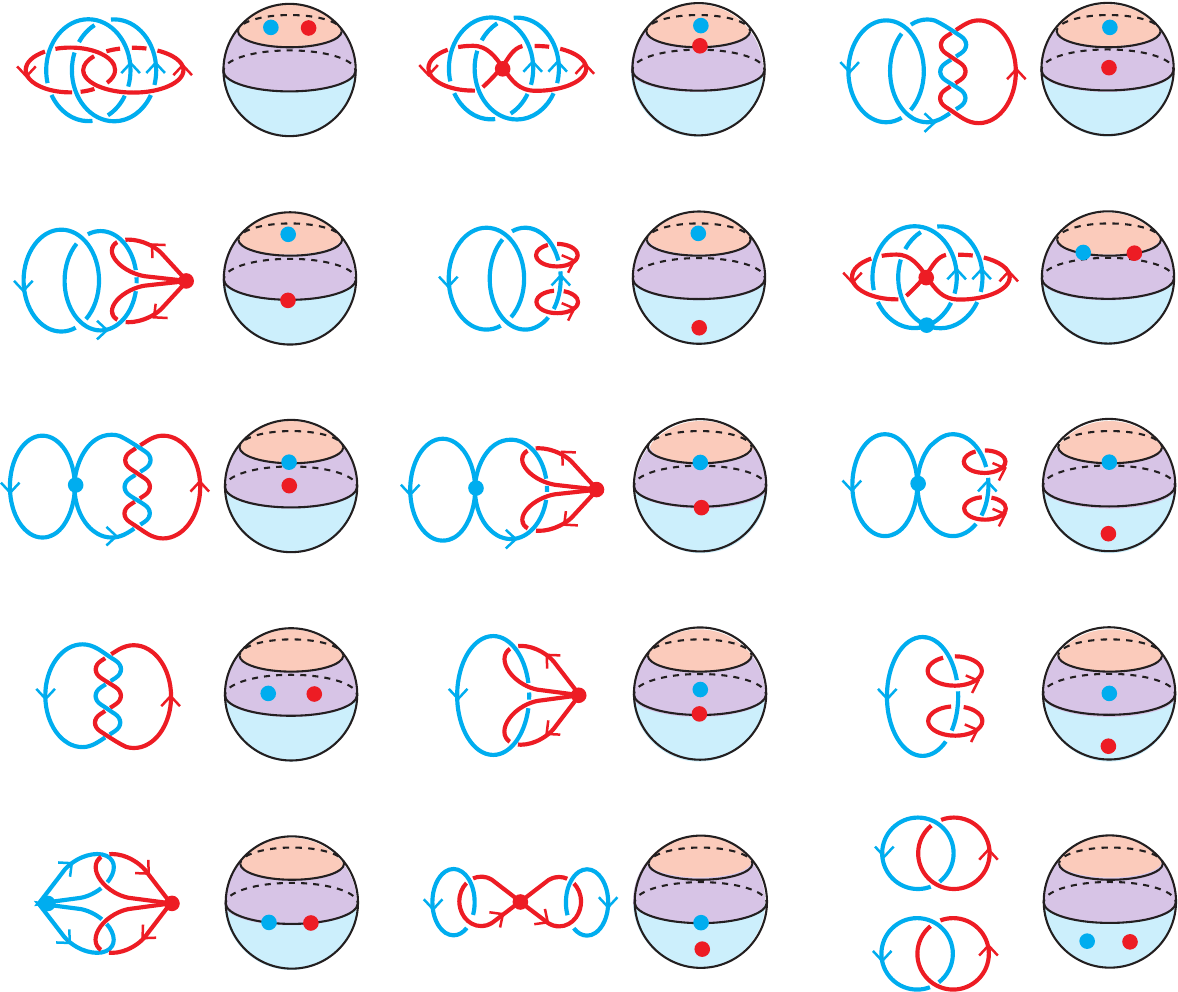}
\caption{All possible spatial graph configurations appearing as preimages of pairs of distinct points in $S^2$ under $\hat{\varphi}_2$. 
This summary of all possible structures explains the observation reported in \cite[Figure~6]{TAS18}, also demonstrating how such findings in the liquid crystal systems can be generalized to other cases of broken axial symmetry of hopfion structures.}
\label{fig:preimage_hat_varphi_2_2}
\end{figure}

\end{document}